\documentclass[journal]{IEEEtran}
\usepackage{cite}
\usepackage{graphicx}
\usepackage{amsmath}
\usepackage{amssymb}
\usepackage{algorithm}

\begin{document}

\title{Artificial-Noise-Aided Message Authentication Codes with Information-Theoretic Security}


%
\author{Xiaofu~Wu, Zhen~Yang, Cong~Ling, and Xiang-Gen Xia
\thanks{This  work was supported in part by the National Natural Science Foundation of China under Grants 61372123, 61271335, by the Key University Science Research Project of Jiangsu Province under Grant 14KJA510003 and by the Scientific Research Foundation of Nanjing University of Posts and Telecommunications under Grant NY213002.}
\thanks{Xiaofu~Wu and Zhen~Yang are with the Key Lab of Ministry of Education in Broadband Wireless Communication and Sensor Network Technology, Nanjing University of Posts and Telecommunications, Nanjing 210003, China (e-mails:
        xfuwu@ieee.org, yangz@njupt.edu.cn)).}
\thanks{Cong Ling is with the Department of Electrical and Electronic Engineering, Imperial College London, London, UK (e-mail: cling@ieee.org).}
\thanks{Xiang-Gen~Xia is with the Department of Electrical and Computer Engineering, University of Delaware, Newark, DE 19716
 (e-mail: xxia@ee.udel.edu).}}



\maketitle

\begin{abstract}
In the past, two main approaches for the purpose of authentication, including information-theoretic authentication codes and complexity-theoretic message authentication codes (MACs), were almost independently developed. In this paper, we propose a new cryptographic primitive, namely, artificial-noise-aided MACs (ANA-MACs), which can be considered as both computationally secure and information-theoretically secure. For ANA-MACs, we introduce artificial noise to interfere with the complexity-theoretic MACs and quantization is further employed to facilitate packet-based transmission. With a channel coding formulation of key recovery in the MACs, the generation of standard authentication tags can be seen as an encoding process for the ensemble of codes, where the shared key between Alice and Bob is considered as the input and the message is used to specify a code from the ensemble of codes. Then, we show that the introduction of artificial noise in ANA-MACs can be well employed to resist the key recovery attack even if the opponent has an unlimited computing power. Finally, a pragmatic approach for the analysis of ANA-MACs is provided, and we show how to balance the three performance metrics, including the completeness error, the false acceptance probability, and the conditional equivocation about the key.  The analysis can be well applied to a class of ANA-MACs, where MACs with Rijndael cipher are employed.
\end{abstract}


\IEEEpeerreviewmaketitle

\section{Introduction}

\PARstart{M}{essage} authentication codes (MACs) are cryptographic primitives used extensively in the construction of
security services, including authentication, nonrepudiation, and integrity. Basically, message authentication is to ensure that an accepted message
truly comes from its acclaimed transmitter. When the transmitter intends to send a message, it also generates a MAC, which is a function of the message and a shared key, known only to both the transmitter and the receiver. The generated MAC is often appended to the message \cite{CyptBook}.  At the receiver, a MAC is computed from the received message and compared to the MAC that is transmitted. If the two MACs are identical, then the transmitter is identified as a legal user and it is highly likely the received message is exactly equal to the one transmitted.

In the past, two main approaches, including information-theoretic authentication codes \cite{MacWilCodes,Simmons} and complexity-theoretic MACs, were almost independently developed for the purpose of authentication. In general, they differ in the assumptions about the capabilities of an opponent. Information-theoretic authentication codes, which are based on information theory, offer unconditional security, i.e., security independent of the computing power of an adversary. The complexity-theoretic approach starts from an abstract model for computation, and assumes that the opponent has limited computing power. Due to their high flexibility, the complexity-theoretic MACs find widespread applications in practice.

Complexity-theoretic MAC algorithms can be constructed from other cryptographic primitives, such as cryptographic hash functions, or block cipher algorithms. Currently, the security of MAC algorithms rely on the hardness of hash functions, i.e, given the message and its MAC, it is ``hard" to forge a MAC on a new message. This means that they can be broken if the adversary has an unlimited power of computation.

In recent years, there has been various efforts \cite{YuIFS,XiaoPHY,DanPHY,TugnaitJSAC} in authenticating the transmitter and receiver at the physical layer, based on prior coordination or secret sharing, where the sender is authenticated if the receiver can successfully demodulate and decode the transmission. In \cite{YuIFS}, a physical-layer authentication scheme was proposed, in which MACs, along with messages, are transmitted concurrently over the physical layer. Compared to the traditional transmission approach above the physical layer, the authors claim the possibility of information-theoretic security due to the presence of channel noise. However, its security often depends on the physical channel.

In this paper, we develop a new cryptographic primitive, artificial-noise-aided MACs (ANA-MACs) for ensuring information-theoretic security. The use of artificial noise in ANA-MACs makes it difficult for an opponent to derive the key. With the use of quantization, ANA-MACs can be encapsulated and transmitted in packets above the physical layer, just like the traditional MACs, which is in sharp contrast to existing physical layer authentication schemes.

It should be pointed out that the proposed ANA-MACs are also different with the binary approximate message authentication codes (AMACs) \cite{AMAC,AMAC02} and the noise-tolerant message authentication codes (NT-MACs) \cite{NTMAC}. Both AMACs and NT-MACs are designed to tolerate some channel errors during the transmission of messages. For ANA-MACs, a slight change in messages may result in a rapid change for authentication tags, as often encountered in the traditional MACs. Yet, ANA-MACs can tolerate some channel errors occurred during the transmission of tags. Furthermore, both AMACs and NT-MACs are computationally secure, while ANA-MACs may ensure some degree of information-theoretic security.

Throughout this paper, we do not discriminate the notations between scalars and vectors, which will be made clear from the contexts. For a binary vector $t$, its bipolar form is simply denoted as $\bar{t}$, in which each component takes value from $\{+1, -1\}$. To be consistent with the standard convention in algorithms and complexity theory, where the running time of an algorithm is measured as a function of the length of its input $n$, we will thus provide the adversary and the honest parties with the security parameter in unary as $1^n$ (i.e., a string of $n$ 1's) when necessary \cite{KatzBook}.

The rest of the paper is organized as follows. Some preliminaries on both information-theoretic authentication codes and MACs are made in Section-II. In Section-III, a new cryptographic primitive, ANA-MACs, is proposed and its verification mechanism is given. Then, its security analysis is formulated in Section-IV. In Section-V, we provide a pragmatic approach for analysis of ANA-MACs. Section-VI presents numerical results and the conclusion is made in Section-VII.

\newtheorem{defn}{Definition}
\newtheorem{lem}{Lemma}
\newtheorem{tem}{Theorem}

\section{Preliminary}
\subsection{Information-Theoretic (Systematic) Authentication Codes}

A systematic authentication code is a triple of $(\mathcal{S},\mathcal{T},\mathcal{K})$ of finite sets and a mapping $E: \mathcal{K} \times \mathcal{S} \rightarrow \mathcal{T}$,  where $\mathcal{S}$ is the source state space, $\mathcal{T}$ is the tag space, $\mathcal{K}$ is the key
space, and $ E(k,\cdot)\triangleq E_k: \mathcal{S} \rightarrow \mathcal{T}$ is often called an encoding rule for a given $k \in \mathcal{K}$.

Two trusting parties, Alice (or a transmitter) and Bob (or a receiver), share a secret key $k \in \mathcal{K}$. To send a piece of information
(called source state) $s \in \mathcal{S}$ to Bob, Alice
computes $t =E_k(s) \in \mathcal{T}$ and puts the message $m=(s,t)$ into a public channel. After receiving $m'=(s',t')$, Bob
will compute $E'_k(s)$ and check whether $t'=E'_k(s)$. If yes, Bob will accept it as authentic. Otherwise, Bob will
reject it.

We assume that an opponent (or Eve) has a complete understanding of the system, including the mapping $E$. The only thing she does not know is the key $k$ agreed upon by Alice and Bob, which is used to specify a particular encoding rule $E_k$.  We also assume that Eve has the ability to introduce a message into the channel. After observation of the first $r$  messages $m_1, \cdots, m_r$, Eve places her own message $m$ into the channel, attempting to make Bob accept it as authentic. This is called a spoofing attack of order $r$. In literature, there are often two different types of spoofing attack, i.e., impersonation attack and substitution attack. An impersonation attack at time $r+1$ \cite{Maurer} is just the spoofing attack of order $r$. In a so-called substitution attack at time $r$, Eve observed  $r$  messages $m_1, \cdots, m_r$ and replaces the message $m_r$ by a different message which she hopes to be accepted by Bob.

For systematic authentication codes, a source state $s$ is assumed to be public (without security) whenever $m=(s,t)$ is transmitted, which can be freely accessed by both Bob and Eve. For this reason, we simply use the authentication tag $t$ instead of a full message $m$ in what follows.

Let $K, T_{r+1}$ and $T^r$ denote the random variables describing the key, the $r+1$-th tag and a sequence of $r$ tags from time 1 to $r$, and taking values $k, t_{r+1}$ and $t^r=(t_1,\cdots,t_r)$, respectively. Let $p_r$ denote the expected probability of successful deception for a spoofing attack of order $r$ and $P_r$ the probability of successful deception if Eve can observe at most $r$ messages.  Walker \cite{Walker} proved
\begin{eqnarray}
\label{splb}
    p_r  \ge 2^{H(K|T^{r+1})-H(K|T^r)}=2^{-I(K;T_{r+1}|T^r)}
\end{eqnarray}
and Rosenbaum \cite{Rosenbaum} proved
\begin{eqnarray}
\label{pflb}
    P_r  \ge 2^{-\frac{1}{r+1}H(K)},
\end{eqnarray}
which hold even if Eve has an unlimited power of computation. If the equality in (\ref{pflb}) holds, the corresponding authentication code is called $r$-perfect.

To prevent Eve from using $t=E_k(s)$ to learn the key $k$, it should have sufficient number of solutions of $k$ for a given $t=E_k(s)$ \cite{MacWilCodes}. Given $s$ and $t$, let $\mathcal{K}(s,t)\triangleq\{k: E_k(s)=t, \forall k \in \mathcal{K}\}$ denotes the set of solutions for  $t=E_k(s)$. It follows that the successful deception probability for a given pair $(s,t)$ has a lower bound of
\begin{eqnarray}
     p_1(s,t) \ge \frac{1}{|\mathcal{K}(s,t)|}.
\end{eqnarray}

In \cite{MacWilCodes}, projective plane codes were proposed to achieve the best possible spoofing attack of order $1$, namely, $P_1=\frac{1}{\sqrt{|\mathcal{K}|}}$.  For the best authentication codes achieving the lower bound of $P_1$, it was proved that
\begin{enumerate}
\item
 $|\mathcal{K}(s_1,t_1) \bigcap \mathcal{K}(s_2,t_2)|=1$ if $s_1 \neq s_2$;
 \item
  $|\mathcal{K}(s,t)|=\sqrt{|\mathcal{K}|}$ for $\forall s\in \mathcal{S}, t \in \mathcal{T}$;
  \item
   $|\{t: |\mathcal{K}(s,t)|>0\}|=\sqrt{|\mathcal{K}|}$ for $\forall s \in \mathcal{S}$.
\end{enumerate}

However, this class of authentication codes cannot resist the spoofing attack of order $r\ge 2$.

\begin{tem}
For 1-perfect systematic authentication codes,  they cannot resist the spoofing attack of order $r\ge 2$, namely,
$p_r=1, \forall r \ge 2$.
\end{tem}
\begin{proof}
It is enough to consider $r=2$. Suppose that Eve has accessed two different messages $m_1=(s_1,t_1)$ and $m_2=(s_2,t_2)$, where $s_1 \neq s_2$. To insert a new message $m=(s,t)$, where $s \neq s_1, s_2$, Eve wants to derive the key $k$, which can be surely learned from two available messages since  $|\mathcal{K}(s_1,t_1) \bigcap \mathcal{K}(s_2,t_2)|=1$. Indeed, there is a single common solution of $k$ for $t_1=E_k(s_1)$ and $t_2=E_k(s_2)$ if $s_1 \neq s_2$.
\end{proof}

Given $s^r=(s_1,\cdots, s_r)$ and $t^r=(t_1,\cdots,t_r)$, let $\mathcal{K}(s^r,t^r)\triangleq\{k: E_k(s_i)=t_i, i=1,\cdots,r, \forall k \in \mathcal{K} \}$ denote the set of solutions for  $t_i=E_k(s_i), i=1,\cdots,r$. Clearly, $\mathcal{K}(s^r,t^r)=\bigcap_{i=1,\cdots,r} \mathcal{K}(s_i,t_i)$.

For $r$-perfect authentication codes, it was shown in \cite{Rosenbaum} that
\begin{eqnarray*}
     |\mathcal{K}(s^i,t^i)|=|\mathcal{K}|^\frac{r+1-i}{r+1}, i \in \{1,\cdots, r+1\}.
\end{eqnarray*}
and  $H(K|T^{r+1})=0$. Hence, we also have the same result as that of Theorem 1 for $r$-perfect authentication codes.

\begin{tem}
For $r$-perfect authentication codes, they cannot resist the spoofing attack of order $l\ge r+1$, namely,
$p_l=1, \forall l \ge r+1$.
\end{tem}

\subsection{Complexity-Theoretic MACs}

\begin{defn}
    A message authentication code $\Pi = (\text{Gen}, \text{Mac}, \text{Vrfy})$ is a triple of algorithms with associated key space $\mathcal{K}$, source message (state) space $\mathcal{S}$\footnotemark\footnotetext{In literature, the message space $\mathcal{M}$ is often used.}, and tag space $\mathcal{T}$.
    \begin{enumerate}
    \item[-] Key Generation. Upon input $1^n$, the algorithm \text{Gen} outputs a uniformly distributed key $k$ of length $n$:  $k \leftarrow   \text{Gen}(1^n)$.

    \item[-] Tagging. The probabilistic authentication algorithm $\text{Mac}_k(s)$ takes as input a secret key $k \in \mathcal{K}$ and a source message $s\in \mathcal{S}$ and outputs an authentication tag  $t \in \mathcal{T}$ .

    \item[-] Verification. The deterministic verification algorithm $\text{Vrfy}_k(s,t)$ takes as input a secret key $k$, a source message $s\in \mathcal{S}$ and a tag $t \in \mathcal{T}$ and outputs an element of the set $\{0, 1\}$.

    \end{enumerate}
\end{defn}

{ \ }

A complexity-theoretic MAC $\Pi = (\text{Gen}, \text{Mac}, \text{Vrfy})$ can be formulated with a keyed hash function.
Formally, the tag is a function of the source message $s$ and the secret key $k$
\begin{eqnarray}
  \label{eq:cn}
     t= \hbar(k,s),
\end{eqnarray}
where $\hbar: \mathcal{K} \times \mathcal{S} \rightarrow \mathcal{T}$ is a keyed hash function.

The verification algorithm takes $k,s,t$ as inputs and outputs a binary decision
\begin{eqnarray}
  \label{eq:cn}
    \nu = \vartheta(k,s,t),
\end{eqnarray}
where $\nu \in \{0,1\}$, and $\vartheta(k,s,t)=1$ if $t = \hbar(s, k)$, zero otherwise.

Note that a MAC implies a two-round authentication protocol: the verifier chooses a random message as challenge, and the prover returns
the MAC on the message.

\begin{defn}[Completeness \cite{KIMAC}] We say that a MAC has completeness error $\alpha$  if for all $s \in \mathcal{S}$ \footnotemark{} \footnotetext{It requires to hold for all $n \in \mathbb{N}$ in \cite{KIMAC} while the completeness error is defined for a given and fixed $n$ in this paper.}
\begin{eqnarray}
  \label{eq:cn}
     P\left[\vartheta(k,s,t)=0 : k\leftarrow \text{Gen}(1^n), t\leftarrow \hbar_k(s) \right]\le \alpha.
\end{eqnarray}
\end{defn}
It is clear that the completeness error $\alpha$ means that the successful authentication probability is larger than $1-\alpha$ for two trusted parties.

\subsection{Remark}
Information-theoretic (systematic) authentication codes provide message authenticity guarantees in an information theoretic sense within a symmetric key setting. However, information theoretic bounds on the spoofing attack of order $r$ show that they are still vulnerable (Theorems 1 and 2) if the opponent can access much more authenticated messages. Complexity-theoretic MACs can be seen as a \textit{counterpart} of information-theoretic authentication codes \textit{in the field of computational security}, without considering the information-theoretic deception probability.

In the past, information-theoretic authentication codes and complexity-theoretic MACs are almost independently developed. It is interesting to ask if we can construct MACs, which are both computationally secure and information-theoretically secure.

\section{Artificial-Noise-Aided MACs}
\subsection{Basic Idea}
We have shown that information-theoretic (systematic) authentication codes take the same function as message authentication codes. If $E(k,s)=\hbar(k,s), \forall k \in \mathcal{K}, s\in \mathcal{S}$,  they are actually the same.

In general, the authentication tag is a deterministic function of a source message $s$ and the key $k$ shared  between Alice and Bob. The only exception is the authentication codes with splitting, where the mapping $\hbar: \mathcal{K} \times \mathcal{S} \rightarrow \mathcal{T}$ is allowed to be stochastic in the sense that, for given $k$ and $s$, $\hbar(k,s)$ is a stochastic variable.

Noting that the use of a stochastic encoding mapping in authentication may be helpful for preventing possible spoofing attacks, since the conditional equivocation about the key may increase compared to a deterministic mapping. In order to make this ideal more
 practical, we propose to introduce artificial noise to corrupt the standard authentication tags.

On one hand, the introduction of artificial noise may increase the conditional equivocation about the key, namely, $H(K|\tilde{T}) \ge H(K|T)$. Here $T$ and $\tilde{T}$ denote the random variables for the standard authentication tag and artificial-noise-corrupted authentication tag, respectively.  On the other hand, the successful authentication probability may decrease as the introduction of noise. Nevertheless, this can be made acceptable in practice if the completeness error is negligible.

\subsection{Formulation}
Suppose that $|\mathcal{T}|=2^l, |\mathcal{K}|=2^n$. To prevent possible eavesdropping, we propose to introduce artificial noise to interfere with the clean tag, and then quantization is used to facilitate packet-based transmission.

An artificial-noise-aided message authentication code (ANA-MAC) is thus with a probabilistic algorithm $\hbar_k^{qw}$ to produce the tag
\begin{equation}
 \tilde{t}  \leftarrow \hbar_k^{qw}(s)
\end{equation}
when the inputs are $k,s$.

In this paper, we consider an explicit construction of $\hbar_k^{qw}(s)$ as follows
\begin{eqnarray}
  \label{eq:anc}
     t &=& \hbar(k, s), \nonumber \\
     \tilde{t} &=& \mathcal{Q}\left(\bar{t} + w\right),
\end{eqnarray}
where $\bar{t}$ is the $l$-length bipolar vector form of $t$, $w$ is an artificially-introduced Gaussian-distributed noise vector with zero mean and variance of $\sigma_w^2 I_l$ ($I_l$ denotes the identity matrix of size $l\times l$), $\mathcal{Q}(x)$ is a $q$-bit quantization function, and $\tilde{t} \in \mathcal{\tilde{T}}$ is a $l$-length vector, where each component takes value from a finite quantization level set $V=\{v_1, v_2, \cdots, v_{2^{q}}\}$ of size $2^q$. Clearly, $\mathcal{\tilde{T}}=V^l$, where $V^l$ denotes the cartesian power of a set $V$.

With the introduction of artificial noise and quantization, the size of an original tag $t$ is expanded by $q$ times. In practice, $q=8$ is often enough.

Given $s$ and $k$, it is possible to partition $\mathcal{\tilde{T}}$ into two disjoint sets, namely, $\mathcal{\tilde{T}}= \mathcal{\tilde{T}}_A(k,s) \cup \mathcal{\tilde{T}}_F(k,s)$, where $P\left(\tilde{t} \in \mathcal{\tilde{T}}_A(k,s)\right) \ge 1-\alpha$. In essence, the verification algorithm for ANA-MAC is to find a deterministic partition of $\mathcal{\tilde{T}}$ for given $s$ and $k$, which minimizes the false acceptance probability and at the same time keeps the successful authentication probability not smaller than a target value of $1-\alpha$.

Whenever such a partition is determined, the verification algorithm can be well formulated. It takes $k,s,\tilde{t}$ as inputs and outputs a binary decision
\begin{eqnarray}
  \label{eq:cn}
    \nu = \vartheta(k,s,\tilde{t}),
\end{eqnarray}
where $\nu \in \{0,1\}$, $\vartheta(k,s,\tilde{t})=1$ if $\tilde{t} \in \mathcal{\tilde{T}}_A(k,s)$ and zero otherwise.

An ANA-MAC has completeness error $\alpha$  if for all $s \in \mathcal{S}$,
\begin{eqnarray}
  \label{eq:cn}
     P\left[\vartheta(k,s,\tilde{t})=0 : k\leftarrow \text{Gen}(1^n), \tilde{t}\leftarrow \hbar_k^{qw}(s) \right]\le \alpha.
\end{eqnarray}

\subsection{Verification with Hypothesis Testing}
\subsubsection{Hypothesis Testing}
Hypothesis testing is the task of deciding which of two hypotheses, $H_0$
or $H_1$, is true, when one is given the value of a
random variable $U$ (e.g., the outcome of a measurement). The
behavior of $U$ is described by two probability distributions: If
$H_0$ or $H_1$ is true, then $U$ is distributed according to the distribution
$p_{H_0}(u)$ or $p_{H_1}(u)$, respectively.

Let $P_D=1-\alpha$ be the detection probability, namely, the probability of successful declaration of $H_0$ when $H_0$ is actually true, and $P_f= \beta$ be the false alarm probability, namely, the probability of false declaration of $H_0$ when $H_1$ is actually true.

The optimal decision rule is given by the famous Neyman-Pearson theorem which states that, for
a given maximal tolerable false alarm probability $\beta$, $\alpha$ can
be minimized by assuming hypothesis if and only if
\begin{eqnarray}
\label{eq:Hyp}
  \log\frac{p_{H_0}(U=u)}{p_{H_1}(U=u)} \ge \varrho
\end{eqnarray}
for some threshold $\varrho$ depending on $\alpha$.

Let the function $\mathcal{D}(\alpha, \beta)$ be defined by
\begin{eqnarray}
  \label{eq:cn}
     \mathcal{D}(\alpha, \beta) = \alpha \log \frac{\alpha}{1-\beta} +  (1-\alpha) \log \frac{1-\alpha}{\beta}.
\end{eqnarray}

With optimal hypothesis testing (\ref{eq:Hyp}), its detection probability and false alarm probability are closely connected \cite{Maurer}.

\begin{lem}
The detection probability $1-\alpha$ and the false alarm probability $\beta$ satisfy
\begin{eqnarray}
  \label{eq:alphaBeta}
     \mathcal{D}(\alpha,\beta) \le D_{KL}\left( p_{H_0}(u)||p_{H_1}(u)  \right)
\end{eqnarray}
where the Kullback-Leibler (KL) divergence can be written as
\begin{eqnarray}
  \label{eq:cn}
      D_{KL}\left( f(x)||g(x) \right) = \sum_x f(x) \log \frac{f(x)}{g(x)}
\end{eqnarray}
for two probability distributions $f(x),g(x)$.
\end{lem}

\subsubsection{Verification}
Now, we focus on the design of verification algorithm for ANA-MACs, which often deals with the impersonation attack. The problem of deciding whether a received tag is authentic or not can be viewed as a hypothesis testing problem \cite{Maurer}.

Let $H_0 $ correspond to the hypothesis that the tag is authentic, and $H_1$ correspond to the hypothesis that the tag has been generated by an adversary. With a standard packet-level transmission above the physical layer, it is assumed that both a legitimate user and an adversary can get a error-free copy of the tag, namely,  $\tilde{t}$.

To facilitate the derivation, we simply assume that $\tilde{t}=\bar{t} + w$, where the quantization is simply omitted. This is a reasonable approximation if a fine quantization method with sufficient number of quantization levels is employed.

To be more concrete, we consider the ``Alice-Bob-Eve'' model, where Eve, as an impersonation attacker, wants to inject messages into the legitimate transmission from Alice to Bob. Suppose Alice and Bob shared a key $k$, which is employed to authenticate each other. With inputs $k,s,\tilde{t}$, Bob wants to decide if $\tilde{t}$ is from Alice. Eve does not know the shared key $k$, and it is assumed that Eve generates a random key $k_E$ for authentication as there is no any information about $k$ available. Essentially, this is cast as a binary hypothesis testing problem:
\begin{eqnarray*}
  \label{eq:hp}
     H_0 &:&  K = k\\
     H_1 &:&  K = k_E.
\end{eqnarray*}

In this case, $U=(\tilde{T}, K)$, $u=(\tilde{t}, k)$. Under hypothesis $H_0$, the pair $u=(\tilde{t}, k)$ (seen by the receiver) is generated according to the distribution $p(\tilde{t}, K=k)$, whereas under hypothesis $H_1$, $u=(\tilde{t}, k)$ is generated according
to the distribution $p(\tilde{t}) \cdot P(K=k)$. This is because that in the case of $H_1$, the generations of authentication tag and key are independent of each other as there is no means to efficiently guess the key.

The formulation of the optimum binary hypothesis testing can be written as
\begin{eqnarray}
\label{eq:opmHyp}
  \eta &=&\log\frac{p_{H_0}(U=u)}{p_{H_1}(U=u)} = \log\frac{p(\tilde{t}, K=k)}{p(\tilde{t}) P(K=k)} \nonumber \\
      &=&\log\frac{p(\tilde{t}|K=k)}{\sum_{k' \in \mathcal{K}} p(\tilde{t}|K=k') P(K=k')}.
\end{eqnarray}
The optimal decision rule is given by $\eta > \varrho$ for some threshold $\varrho$ depending on $\alpha$.

Now, it is clear that a partition of $\tilde{\mathcal{T}}$ for the purpose of verification can be done as
\begin{equation}
   \mathcal{\tilde{T}}_A(k,s) = \left\{\tilde{t} \in \tilde{\mathcal{T}}: \eta > \varrho \right\}.
\end{equation}

As the source message $s$ is assumed to be available, it follows that
\begin{eqnarray}
  p(\tilde{t}|k) \propto \exp\left[-\frac{(\tilde{t}-\bar{t})^T (\tilde{t}-\bar{t})}{2\sigma_w^2}\right]
\end{eqnarray}
with $t=\hbar(k,s)$ and $\bar{t}$ is its bipolar (column) vector form.

In general, this binary hypothesis testing problem in its optimum form can not be easily tackled as it requires to enumerate $2^n$ keys with a priori uniform distribution.

As the optimum hypothesis testing is difficult to implement, we propose to use a simple test statistic
\begin{eqnarray}
  \label{eq:decMetric}
   \eta = \bar{\mu}^T \tilde{t},
\end{eqnarray}
and $\eta$ is further compared to a threshold value $\varrho$ for making a final decision, where $\mu = \hbar(k,s)$ is the tag generated by Bob and $\bar{\mu}$ is its bipolar vector form.

This approach can be viewed as a code acquisition approach encountered in code-division multiple-access (CDMA) communication systems, where the tag signature $\mu$ can be considered as a unique pseudo-noise (PN) code, which is available at the sides of both Alice and Bob, but keeps unknown to any potential attacker.

In both hypotheses, $\eta$ is the sum of $l$ normally distributed random variables, which is still normally distributed. Therefore, it suffices to compute its mean and variance.

In the case of hypothesis $H_0$, one can show that
\begin{eqnarray}
  \label{eq:cn}
     \eta|H_0 = l + z_0,
\end{eqnarray}
where $z_0 = \sum_{i=1}^{l} \bar{\mu}_i w_i$. We denote its mean and variance as
\begin{eqnarray}
  \label{eq:cn}
     \bar{\eta}_0 &\triangleq& E\{\eta|H_0\} = l, \nonumber \\
     \sigma^2_{H_0}&\triangleq&  \text{Var}\{\eta|H_0\} = l \sigma_w^2.
\end{eqnarray}

By decomposing the hypothesis $H_1$ into a series of sub-hypothesises $\left\{H_1^{k'}: H_1, K=k'\right\}$, i.e.,  by further assuming that Eve impersonates Alice using the key $k'$,  we have
\begin{eqnarray}
  \label{eq:etaH1}
     \eta|H_1^{k'} = l- 2 d_H\left(\hbar(k, s),\hbar(k', s)\right) + z_1,
\end{eqnarray}
where $z_1 = \sum_{i=1}^{l} \bar{\mu}_i w_i$ and $d_H(x,y)$ denotes the Hamming distance between two binary strings of $x$ and $y$. Then,
\setlength{\arraycolsep}{0.0em}
\begin{eqnarray}
  \label{eq:staH1}
     \bar{\eta}_1^{k'} &\triangleq& E\{\eta|H_1, k'\} = l-2 d_H\left(\hbar(k,s),\hbar(k',s)\right), \nonumber \\
     \sigma^2_{H_1^{k'}}&\triangleq& \text{Var}\{\eta|H_1, k'\} = l \sigma_w^2.
\end{eqnarray}
\setlength{\arraycolsep}{5pt}
It is clear that $\eta|H_0 \sim \mathcal{N}\left(\bar{\eta}_0,\sigma^2_{H_0}\right)$ and  $\eta|H_1^{k'} \sim \mathcal{N}\left(\bar{\eta}_1^{k'},\sigma^2_{H_1^{k'}}\right)$.

The authentication is typically claimed if $\eta \ge \varrho$. Hence, the successful authentication probability (or the detection probability) can be simply computed as
\begin{eqnarray}
  \label{eq:PD}
        P_D = Q\left(\frac{\varrho - \bar{\eta}_0}{\sigma_{H_0}}\right),
\end{eqnarray}
where
\begin{equation}
 \label{eq:Q}
 Q(x)=\frac{1}{\sqrt{2\pi}}\int_x^{\infty}
 \exp\left(-\frac{t^2}{2}\right)dt.
\end{equation}

With this setting of threshold $\varrho$, according to the distribution of $\eta|H_1$, a false alarm
probability $\beta$ can be calculated as
\begin{eqnarray}
 \label{eq:beta}
    \beta =  E_{k'}\left[Q\left(\frac{ \varrho-\bar{\eta}_1^{k'}}{\sigma_{H_1^{k'}}}\right)\right].
\end{eqnarray}

We comment here that the successful authentication probability (\ref{eq:PD}) can be directly computed while the false alarm probability is difficult to compute in general, since it should enumerate all possible keys, which is of size $2^n$. Furthermore, the above formulation in general depends on the source message $s$ as show by (\ref{eq:staH1}). Indeed, one should enumerate all keys to compute the false alarm probability for any given $s\in \mathcal{S}$. This seems to be an impossible task. Later in Section-V, we, however, show that it is possible to compute it in a closed-form expression thanks to the pseudorandomness of the complexity-theoretic MACs.

\section{Security Analysis}
\subsection{Information-Theoretic Bounds}

Consider an impersonation attack on the $(r+1)$th source message $s_{r+1}$. We adopt the powerful hypothesis-testing formulation originally proposed by Maurer \cite{Maurer}. The receiver knows $K$ and $r$ messages $m_1=(s_1,\tilde{t}_1), \cdots, m_r=(s_r,\tilde{t}_r)$, and sees a message $m_{r+1}=(s_{r+1},\tilde{t}_{r+1})$ which could either be a correct message sent by Alice (hypothesis $H_0$) or a fraudulent message inserted by Eve (hypothesis $H_1$).

For this spoofing attacker of order $r$, the opponent's strategy for impersonation at time $r+1$ can be described by an arbitrary probability
distribution \cite{Maurer} $Q_{M_{r+1}=m_{r+1}|M_1=m_1,\cdots, M_r=m_r}$. If the opponent chooses to use $Q_{M_{r+1}=m_{r+1}|M_1=m_1,\cdots, M_r=m_r}=P_{M_{r+1}=m_{r+1}|M_1=m_1,\cdots, M_r=m_r}$, the cheating probability has the following lower bound.

\begin{tem}
Consider the spoofing attack of order $r$ for an ANA-MAC, where the opponent generates an ANA-MAC tag $\tilde{T}_{r+1}$ when she/he observed $r$ ANA-MAC tags ($\tilde{T}^r$). We have
\begin{eqnarray}
\label{eq:fl}
  \mathcal{D} \left(\alpha, p_r\right) \le I\left(K; \tilde{T}_{r+1}|\tilde{T}_1, \cdots, \tilde{T}_{r}\right),
\end{eqnarray}
and for $\alpha=0$,
\begin{eqnarray}
\label{eq:pr}
    p_r  \ge 2^{-I\left(K; \tilde{T}_{r+1}|\tilde{T}_1, \cdots, \tilde{T}_{r}\right)}.
\end{eqnarray}
\end{tem}
\begin{proof}
 Consider probability distributions conditioned on the event that $M_1=m_1, \cdots, M_r=m_r$. Under hypothesis $H_0$, the pair $U=[M_{r+1},K]$ (seen by the receiver) is generated according
to the probability distribution
\begin{equation*}
 P_{M_{r+1}, K|M_1=m_1,\cdots,M_r=m_r},
\end{equation*}
whereas under hypothesis $H_1$, $U=[M_{r+1},K]$ is generated according
to the distribution
\begin{equation*}
 P_{M_{r+1}=m_{r+1}|M_1=m_1,\cdots, M_r=m_r} \cdot P_{K|M_1=m_1,\cdots,M_r=m_r}.
\end{equation*}

For ANA-MACs, each message $m$ can be written as $m=(s,\tilde{t})$, where the source message $s$ is carried without secrecy and hence is accessible even for any opponent. Hence, we can employ a more compact form for probability distributions, namely,
\begin{eqnarray*}
 P_{M_{r+1}, K|M_1=m_1,\cdots,M_r=m_r} &=& P_{\tilde{T}_{r+1}, K|\tilde{T}_1=\tilde{t}_1,\cdots,\tilde{T}_r=\tilde{t}_r}, \\
 P_{M_{r+1}=m_{r+1}|M_1=m_1,\cdots, M_r=m_r} &=&  P_{\tilde{T}_{r+1}|\tilde{T}_1=\tilde{t}_1,\cdots,\tilde{T}_r=\tilde{t}_r}, \\
 P_{K|M_1=m_1,\cdots,M_r=m_r} &=& P_{K|\tilde{T}_1=\tilde{t}_1,\cdots,\tilde{T}_r=\tilde{t}_r}.
\end{eqnarray*}

Let $p_r(\tilde{t}_1,\cdots,\tilde{t}_r)$ denote the successful deception probability for a particular observed sequence $\tilde{T}_1=\tilde{t}_1,\cdots, \tilde{T}_r=\tilde{t}_r$, which is the probability of accepting hypothesis $H_1$ when $H_0$ is actually true. According to Lemma 1, we have
\begin{eqnarray*}
\mathcal{D} \left(\alpha, p_r(\tilde{t}_1,\cdots,\tilde{t}_r)\right) \le I\left(K; \tilde{T}_{r+1}|\tilde{T}_1=\tilde{t}_1,\cdots,\tilde{T}_r=\tilde{t}_r\right).
\end{eqnarray*}
Then, it is straightforward to show both (\ref{eq:fl}) and (\ref{eq:pr}) just did in \cite{Maurer}.
\end{proof}

To gain further insights into the spoofing attack, we can also follow the derivation process employed in \cite{Rosenbaum} \cite{Pei}.

Within the framework of ANA-MACs, we argue that the opponent should do her/his best to generate a clear authentication tag $t$, instead of a noise-corrupted version $\tilde{t}$, given that $\tilde{t}^r$ has been observed for a spoofing attack of order $r$. Indeed, if an illegal tag is generated by the opponent,  the introduction of the artificial noise may slightly increase the false acceptance probability. However, this increase is often minor as the false acceptance probability should be less than a small target value in the design of ANA-MACs.

\begin{tem}
Consider the spoofing attack of order $r$ for an ANA-MAC, where the opponent generates a clear tag $T_{r+1}$ when she/he observed $r$ ANA-MAC tags ($\tilde{T}^r$).  We have
\begin{eqnarray}
    p_r  \ge 2^{-I\left(K; T_{r+1}|\tilde{T}_1, \cdots, \tilde{T}_{r}\right)}.
\end{eqnarray}
\end{tem}

\begin{proof}
Let $P_r(t|\tilde{t}^r)$ denote the probability that $t$ would be a valid choice for $\tilde{T}_{r+1}$
given that $\tilde{T}^r=\tilde{t}^r$ has been observed. Then,
\begin{eqnarray}
\label{eq:cheat}
    P_r(t|\tilde{t}^r) &=& \sum_{k\in \mathcal{K}} P(t,k|\tilde{t}^r) =  \sum_{k\in \mathcal{K}} P(t|k,\tilde{t}^r)P(k|\tilde{t}^r)  \nonumber \\
    &=&  \sum_{k\in \mathcal{K}(t)} P(k|\tilde{t}^r),
\end{eqnarray}
where $\mathcal{K}(t)$ is the set of keys under which $t$ is a valid tag.

Given that $\tilde{t}^r$ has been observed, the opponent's optimum strategy is to substitute the tag $t$
that maximizes  $P_r(t|\tilde{t}^r)$. Thus, the success  probability given that $\tilde{t}^r$ has been observed in an
optimum spoofing attack of order $r$ is
\begin{eqnarray}
    P_r(\tilde{t}^r) &\triangleq& \max_{t \in \mathcal{T}} P_r(t|\tilde{t}^r) \nonumber \\
     &\ge& \sum_{t \in \mathcal{T}} P(T_{r+1}=t|\tilde{t}^r)P_r(t|\tilde{t}^r) \nonumber \\
    &=& \sum_{t \in \mathcal{T}} \sum_{k\in \mathcal{K}(t)} P(T_{r+1}=t|\tilde{t}^r)P(k|\tilde{t}^r)  \nonumber \\
    &=& E\left\{ \frac{P(T_{r+1}=t|\tilde{t}^r)P(k|\tilde{t}^r)}{P(T_{r+1}=t,k|\tilde{t}^r)}\right\}
\end{eqnarray}
where \textit{E} is the conditional expectation given that $\tilde{T}^r=\tilde{t}^r$.

By use of Jensen's inequality, we have
\begin{eqnarray}
    P_r(\tilde{t}^r) &\ge& 2^{H(K,T_{r+1}|\tilde{T}^r=\tilde{t}^r)- H(T_{r+1}|\tilde{T}^r=\tilde{t}^r) -H(K|\tilde{T}^r=\tilde{t}^r)} \nonumber \\
     &=& 2^{-I(K,T_{r+1}|\tilde{T}^r=\tilde{t}^r)}.
\end{eqnarray}
\end{proof}

As shown in (\ref{eq:cheat}), the conditional cheating probability is determined by the opponent's capability to compute the a posterior probabilities about the key when she/he observed $r$ tags, namely, $P(k|\tilde{t}^r), \forall k\in \mathcal{K}$. Therefore, it is interesting to develop a coding formulation for the problem of key recovery in ANA-MACs.

\subsection{A Coding Formulation for Key Recovery in MACs}
Consider the key recovery problem for the spoofing attack of order $r$, namely, the opponent has accessed $r$ messages $m_1=(s_1,t_1), \cdots, m_r=(s_r,t_r)$ and he/she wants to recover the key. Now, we present a coding formulation for this problem.

In the opponent's view (for key recovery), the generation of possible tags for a given message $s$ can be considered as a deterministic encoding process of
\begin{eqnarray}
  \hbar(\cdot, s): \mathcal{K} \rightarrow \mathcal{T}.
\end{eqnarray}

Given $r$ source messages $s^r$, the generation of possible $r$-tags is with a determinist encoding process of
 \begin{eqnarray}
  \hbar(\cdot, s^r)\triangleq \left[\hbar(\cdot, s_1), \cdots, \hbar(\cdot,s_r)\right]: \mathcal{K} \rightarrow \mathcal{T}^r.
\end{eqnarray}

That means, given $r$ source messages $s^r=(s_1,\cdots,s_r)\in \mathcal{S}^r$, it is possible to generate a code $\mathcal{C}(s^r)$, which is comprised of $|\mathcal{K}|=2^{n}$ codewords, namely,
\begin{eqnarray}
\mathcal{C}(s^r)=\{c_1(s^r), \cdots, c_{2^n}(s^r)\},
\end{eqnarray}
where each codeword $c_k(s^r)= \left(\hbar\left(k,s_1\right), \cdots, \hbar\left(k,s_r\right)\right)$ is indexed by a possible key $k\in \mathcal{K}$.

In what follows, we say $\mathcal{C}(s^r)$ as an $r$-order MAC, corresponding to the spoofing attack of order $r$.

Clearly, there are $|\mathcal{K}|=2^{n}$ codewords. Suppose that the cardinality of tag space is $|\mathcal{T}|=2^l$ and each tag is of the equal binary bit length $l$, the coding rate of $\mathcal{C}(s^r)$ can be defined as
\begin{equation}
  R_c(r) = \frac{n}{rl}.
\end{equation}

Since the source message $s$ is generated according to a finite message set $\mathcal{S}$, the opponent has to consider an ensemble of codes $\Omega_r(\mathcal{C})=\{\mathcal{C}(s^r): s^r \in \mathcal{S}^r\}$, which is all of fixed coding rate $R_c(r)$.

This ensemble of codes $\Omega_r(\mathcal{C})$ is revealed to both Alice and Bob. From a standard cryptographic view, this code ensemble is also revealed to Eve.

In the literature, the size of tag space is often not larger than the size of key, which yields $R_c(1) \ge 1$. For information-theoretic authentication codes, it is always assumed that $R_c(r) > 1$ for some $r$'s. Otherwise, it is not secure. For the MACs encountered in practice,  $R_c(1)\ge 1$. However, $R_c(r) \le 1$ typically for $r\ge 2$. For example, the 3GPP employs a challenge-response authentication scheme, where the binary length of a tag is $l=64$, while the binary length of a key is $n=128$.

According to the value of coding rate $R_c(r)$, it can be formulated as either a source coding problem ($R_c(r)> 1$) or a channel coding problem ($R_c(r) \le 1$) for key recovery in MACs.

In \cite{Andrea}, the link between authentication theory and rate-distortion theory was exploited and the rate-distortion function appears in a powerful lower bound to the probability of an authentication fraud. In essence, Sgarro introduced a binary fraud matrix, which tells which authenticated tags cheat which keys under the given attack: $\chi(k, t)=1$ iff the authenticated message $m=(s,t)$ cheats the key $k$. The distortion between $k$ and $t$ can be defined as a complement form of $\chi(k, t)$, namely, $d(k,t) = 1-\chi(k,t)$. Positive distortion levels $\Delta>0$ make sense in a situation when the legal user is recognized as such whenever a ``sufficiently high fraction" of the received tags are authenticated.

It should be pointed out that Sgarro in \cite{Andrea} considered only the spoofing attack of order 1 by a careful definition of the fraud matrix.
For the spoofing attack of order $r$, the distortion between $k$ and $t^r$ should be defined as a complement form of $\chi(k, t^r)$, namely,
\begin{equation}
   d(k,t^r) = 1-\chi(k,t^r).
\end{equation}
The rate-distortion function for the ``key source'' $K$ with probability distribution $\pi$ (often uniform) and distortion measure $d(k,t^r)$ is defined as
 \begin{equation}
   R(\Delta) = \min_{P_K=\pi,E\{d(K,T^r)\}<\Delta} I(K;T^r).
\end{equation}

For any opponent who observed $r$ messages $m_1=(s_1,t_1),\cdots, m_r=(s_r,t_r)$,  his/her equivocation about the key is upper bounded by
\begin{eqnarray}
  H( K |T^r ) \le H(K)- R(\Delta=0),
\end{eqnarray}
where the rate-distortion function  $R(\Delta)$ can be numerically computed.

In what follows, we mainly focus on the channel coding formulation, as this will eventually be the case ($R_c(r)\le 1$) for some $r$'s when the opponent can access $r$ (different) authentication tags. We point out that even in the case of $r=1$, it is also possible to construct authentication tags with $l \ge n$ \cite{Daemen:2002:DRA}. The expanded size of tag space can be well employed to enhance the receiver operating characteristic (ROC) performance for authentication, which, however, is more vulnerable to potential attackers.  This vulnerability can be remedied by the introduction of artificial noise in ANA-MACs.

\subsection{A Decoding Approach for Key Recovery in ANA-MACs}
For an ANA-MAC under the spoofing attack of order $r$, we can characterize it using a quintuple $\left\{\mathcal{S},\mathcal{K}, \mathcal{T}, \Omega_r(\mathcal{C}), p(y|x)\right\}$, where $p(y|x)$ denotes the conditional probability distribution for  the artificially-introduced channel between $\bar{t}$ ($x$) and $\tilde{t}$ ($y$).  In this paper, we always assume a memoryless channel and hence, $p(y|x)=\prod_{i=1}^{rl} p(y_i|x_i)$.

Firstly, we consider the transmission of MACs, in which Eve can directly access the $r$ source messages $s^r$ and their tags
\begin{equation*}
 y=\hbar(k,s^r)\triangleq [\hbar(k,s_1), \cdots,\hbar(k,s_r)].
\end{equation*}

Given $s^r$ and if the encoding rule
\begin{eqnarray*}
  \hbar(\cdot, s^r): \mathcal{K} \rightarrow \mathcal{T}^r
\end{eqnarray*}
is an injection ($R_c(r)\le 1$), Eve can recover the key $k$ by generating a lookup table of size $2^{n}$ and searching over this table for finding the key $k$, which admits $y=\hbar(k,s^r)$.

In the language of coding, it means that the recovery of key can be considered as decoding of the received signal $Y$ to its most likely input $\hat{K}(Y)$. Given $r$ messages $m_1,\cdots, m_r$, if any decoder $\hat{K}(Y)$ is of computational complexity $\mathcal{O}(2^n)$, we claim that the computational security can be achieved for this message authentication code.

For ensuring computational security, it requires that no any efficient decoding algorithm exists for any code $\mathcal{C}(s^r) \in \Omega_r(\mathcal{C})$. Since the publication of Shannon's original paper in 1948, the search of the codes for achieving the channel capacity has been pursued for several decades. Currently, linear codes and their efficient decoding algorithms have been extensively studied. Therefore, for construction of a good ANA-MAC code, linear code ensembles should be better avoided as their complexity can often be reduced due to the linearity of codes. As various hash functions are nonlinear, this is practically avoided for the construction of MACs based on the keyed hash functions.

{\ }

To derive an explicit key for the spoofing attack of order $r$, it is best to use a maximum-likelihood decoder for ANA-MACs if the adversary has unlimited computing power.
\begin{defn}
Let the binary codeword $c\in C$, which is further modulated with $x(c)$ and transmitted over the channel $p(y|x)$, the received vector $y \in \mathcal{R}^{rl}$. A maximum-likelihood (ML) decoding algorithm decodes
the vector $y$ into a codeword $\hat{c}$, such that
\begin{eqnarray}
  \label{eq:cn}
     \hat{c} = \max_{c \in \mathcal{C}} p\left(y|x(c)\right).
\end{eqnarray}
\end{defn}

\begin{defn} (ML recoverable) Given $y \in \mathcal{R}^{rl}$ and $s^r$, where $y= x + w$ and $x=\bar{c}, c=\hbar(k,s^r)$. For an ML decoder $\hat{k}(y)$,  we mean that
\begin{eqnarray}
  \label{eq:cn}
     \hat{k} = \max_{k \in \mathcal{K}} p(y|k,s^r).
\end{eqnarray}
If $P(\hat{k} \neq k)=0$, we claim that the authentication key is ML recoverable.
\end{defn}
{\ }

We consider a binary-input continuous-output AWGN channel (Bi-AWGN) as encountered in ANA-MACs (\ref{eq:anc}). Its capacity $C_2\left(\gamma_t\right)$ is a function of $\gamma_t=1/2\sigma_w^2$, which can be explicitly expressed as
\begin{eqnarray*}
       C_2 (\gamma_t) = \left[1 - \frac{1}{\sqrt{2\pi}} \int_{-\infty}^{\infty} e^{-(y-\beta)^2/2} \log_2 \left(1 + e^{-2\beta y}\right)dy\right],
\end{eqnarray*}
where $\beta=\sqrt{2 \gamma_t}$. As the value of $\gamma_t$ is determined by the introduced artificial noise, one can adjust it in practice for the best possible performance.

The sphere-packing bound of Shannon \cite{SPB1959} provides a lower bound on the decoding error probability of block codes transmitted over the Bi-AWGN channel. With a coding approach for MACs, the best possible recovery of key for a potential eavesdropper to attack ANA-MACs is to use an ML decoder, with which, the decoding probability can be lower bounded with the Shannon's 1959 sphere-packing bound.

\begin{lem} (The SP59 Lower Bound \cite{SPB1959}) Consider an $r$-order ANA-MAC code $\left\{\mathcal{S}^r,\mathcal{T}^r,\mathcal{K},\Omega_r(\mathcal{C}), p(y|x)\right\}$. Let a sequence of source messages $s^r \in \mathcal{S}^r$ be sent, and $p(y|x)$ represents a Bi-AWGN channel with the signal-to-noise ratio of $\gamma_t$. For any decoder $\hat{K}$, it is clear that $K\rightarrow \hbar(K,s^r) \rightarrow X \rightarrow Y \rightarrow \hat{K}$ form a Markov process. Let $P_e = P(K \neq \hat{K})$, we have that
\begin{eqnarray*}
  \label{eq:cn}
   P_e > P_{SPB}\left(l,\theta,\gamma_t\right),
\end{eqnarray*}
where
 \begin{eqnarray*}
      P_{SPB}\left(l,\theta,\gamma_t\right) = Q(\sqrt{2l\gamma_t}) + \frac{l-1}{\sqrt{2\pi}}e^{-l \gamma_t} \nonumber \\
         \cdot \int_{\theta}^{\pi/2}\sin(\phi)^{l-2} f_{l}(\sqrt{2l\gamma_t}\cos(\phi))d\phi,
\end{eqnarray*}
\begin{equation*}
   f_l(x)=\frac{1}{2^{\frac{l-1}{2}} \Gamma\left(\frac{l+1}{2}\right)} \int_0^\infty z^{l-1}\exp\left(-\frac{z^2}{2}+zx\right)dz,
\end{equation*}
and $\theta \in [0,\pi]$ satisfies the inequality  $2^{-l R} \le \frac{\Omega_{l}(\theta)}{\Omega_{l}(\pi)}$ with
\begin{eqnarray*}
   \Omega_{l}(\theta) = \frac{2\pi^{\frac{{l}-1}{2}}}{\Gamma(\frac{l-1}{2})} \int_0^\theta (\sin(\phi))^{l-2} d\phi.
\end{eqnarray*}

\end{lem}
{\ }

The SP59 bound is exponentially increasing with the block length $l$ and the exponent is strictly negative for all $R_c(r) \triangleq \frac{n}{rl}> C_2(\gamma_t)$, it becomes clear that above capacity the minimum probability of error goes to 1 exponentially fast with the block length. Hence, any opponent cannot recover the key explicitly  for a properly-designed ANA-MAC, as summarized as follows.

\begin{tem} Given an $r$-order ANA-MAC $\left\{\mathcal{S}^r,\mathcal{T}^r,\mathcal{K},\Omega_r(\mathcal{C}), p(y|x)\right\}$. With an artificially-introduced Bi-AWGN channel of noise variance $1/2\gamma_t$, we say that this ANA-MAC can resist any explicit key-recovery attack as the recovery of key is with error probability \textit{exponentially} approaching 1 even for any adversary with an unlimited power of computation if $R_c(r) > C_2(\gamma_t)$ when $l \rightarrow \infty$.
\end{tem}
{\ }

In practice, the key is often of short length, typically of length 128. Hence, it seems that Theorem 5 makes no sense. Fortunately, it is well known in coding theory that the decoding error probability can go to 1 even with short block length (\textit{exponentially}) if the signal-to-noise ratio $\gamma_t$ is sufficiently low, which is implied by the SP59 lower bound. We'll show numerical results later.

As shown in (\ref{eq:cheat}), the conditional cheating probability is determined by the opponent's capability to compute the a posterior probabilities, $P(k|\tilde{t}^r), \forall k\in \mathcal{K}$. Therefore, it is more fundamental to derive a lower bound on the conditional equivocation about the key $H( K | \tilde{T}^r )$ when the opponent has accessed $r$ tags.

\begin{tem} (Lower Bound on the Conditional Equivocation about the Key)
For any adversary who has observed $r$ ANA-MAC pairs of $(s_i, \tilde{t}_i), i=1,\cdots,r$,  her/his equivocation about the key is lower bounded by
\begin{eqnarray}
  H( K | \tilde{T}^r ) \ge n\left(1-{R_c(r)}^{-1}  C_2(\gamma_t)\right),
\end{eqnarray}
where $n$ is the key length and $\gamma_t$ is the SNR due to the introduction of artificial noise in ANA-MACs.
\end{tem}
\begin{proof}
As the mutual information per channel use between the observation at the side of Eve and the shared key $\frac{1}{n}I(K; \tilde{T}^r)$ is upper bounded by the channel capacity, his/her equivocation about $K$ when Eve observed various realizations of $\tilde{T}^r$ can be lower bounded as
\begin{eqnarray}
  H( K | \tilde{T}^r ) &=& H(K) - I(K; \tilde{T}^r),  \nonumber \\
                     &\ge& H(K) - n {R_c(r)}^{-1} C_2(\gamma_t) \nonumber \\
                     &=& n\left(1-{R_c(r)}^{-1}  C_2(\gamma_t)\right).
\end{eqnarray}
\end{proof}

 Let
 \begin{equation}
    \label{eq:lowb}
    \delta = 1- {R_c(r)}^{-1} C_2(\gamma_t),
 \end{equation}
it follows that $H(K |\tilde{T}^r) \ge \delta H(K)$. Hence, the successful probability for an eavesdropper to guess the key is about $2^{-\delta n}$.

Clearly, $\delta$ is a lower bound on the normalized equivocation (relative to the entropy of key).

\section{A Pragmatic Approach for the Analysis of ANA-MACs}
For the design of ANA-MACs, one should carefully balance the three performance metrics, namely, the successful authentication probability, the false acceptance probability and the security against spoofing attacks. For simplicity, we focus on the spoofing attack of order-1 and $R_c(1) \le 1$, in which a channel coding formulation makes sense.

As shown in (\ref{eq:cheat}), the conditional cheating probability is determined by the opponent's capability to compute $P(k|\tilde{T}=\tilde{t}), \forall k\in \mathcal{K}$. Hence, a tractable metric for the security against spoofing attacks can be chosen to be the conditional equivocation about the key $H(K|\tilde{T})$.

With a channel coding formulation for MACs, we now show that it is possible to provide a design guideline for balancing the three performance metrics of ANA-MACs. We start with a brief review of some basic concepts of channel coding.

A binary $(l, M, d)$ code represents a binary code with length $l$, size $M = |C|$, and minimum Hamming distance $d$. An equidistant code (of length $l$ and distance $d$) is a set $C$ of vectors of length $l$ (called codewords), such that $d(x, y)=d$ for all distinct $x, y \in C$.

The distance distribution of a binary code $C$ of length $l$ is defined to be the $(l + 1)$-tuple $(A_0(C), A_1(C),... , A_l(C))$,
where $A_i(C)$ denotes the mean number of codewords at Hamming distance $i$ from a fixed codeword.

A code $C$ is said to be distance invariant if the number of codewords at distance $i$ from a fixed codeword only depends on $i$ and not on the particular word chosen.

Given $s\in \mathcal{S}$ and an ANA-MAC, let us first suppose that the underlying MAC $C \triangleq \mathcal{C}(s)$ is an equidistant code.

\subsection{Equidistant MACs}
\begin{lem}[Semakov and Zinoviev \cite{Semakov}]
An optimal binary equidistant $(l, M, d)$ code exists if and only if there exists
a resolvable balanced incomplete block design (BIBD) with parameters $v = M, k = M/2, \lambda = l-d, r = l$.
\end{lem}

For binary equidistant $(l, M, d)$ code, the distance takes the value of
\begin{eqnarray}
   d_{opt}=\frac{Ml}{2(M-1)}=\frac{l+1}{2}
\end{eqnarray}
if $d_{opt}$ is an integer. If $d_{opt}$ is not an integer, i.e. the equidistant code is not optimal, then the code with
$d =\lfloor d_{opt}\rfloor$ is called a good equidistant code. Some constructions of good equidistant codes from balanced arrays and nested BIBDs were described in \cite{Sinha}.

Suppose now that the underlying MACs employed in ANA-MACs are ($l, 2^n, d$) equidistant codes. Then, it is possible to compute the three performance metrics.

Firstly, the use of equidistant MACs can facilitate the computation of the successful authentication probability $1-\alpha$ and the false acceptance probability $\beta$. According to the decision metric of (\ref{eq:decMetric}) and further setting
\begin{equation*}
   \varrho=\rho l,
\end{equation*}
it follows that
\begin{eqnarray}
\label{eq:alpha-ed}
   \alpha &=& 1 - P_D =  Q\left(\frac{\bar{\eta}_0-\varrho}{\sigma_{H_0}}\right) = Q\left(\sqrt{2\gamma_t l}(1-\rho)\right) \nonumber \\
    &=&  Q\left(\sqrt{2\gamma_b n}(1-\rho)\right)  \nonumber \\
    &\triangleq& Q\left(\sqrt{2\gamma_b G}\right)
\end{eqnarray}
where $G \triangleq (1-\rho)^2$, $\gamma_b \triangleq R_c^{-1} \gamma_t$ and
\begin{eqnarray}
\label{eq:beta-ed}
   \beta &=&  \Pr\left(\eta = \bar{\mu}^T \tilde{t} \ge \rho l\right) \nonumber \\
                &=&  Q\left(\sqrt{2\gamma_t l}\left(2\delta_d-(1-\rho)\right)\right) \nonumber \\
                &=&  Q\left(\sqrt{2\gamma_b n}\left(2\delta_d-(1-\rho)\right)\right).
\end{eqnarray}

For example, Let us consider the special case of $\beta = \alpha$. According to (\ref{eq:alpha-ed}) and (\ref{eq:beta-ed}), this means that
\begin{equation}
\label{eq:d}
   \delta_{d,l} \triangleq \frac{d}{l}= \sqrt{\frac{G}{n}}.
\end{equation}

The conditional equivocation about the key $H(K|\tilde{T})$ can be well evaluated by the lower bound proposed in Theorem 6. For equidistant MACs, we can provide a heuristic approximation method to evaluate it, which shows an explicit connection between $H(K|\tilde{T})$ and $d$.

\begin{tem}
For an ANA-MAC with the use of $(l, 2^n, d)$ equidistant codes for the underlying MACs, the conditional equivocation about the key when the opponent has accessed a single tag can be approximated as
\begin{eqnarray}
\label{eq:happrox}
    H(K|\tilde{T}) \approx  n - 4 \ln(2)^{-1} R_c \gamma_b d .
\end{eqnarray}
\end{tem}
\begin{proof}
 Consider that a secret key $k$ shared between Alice and Bob is used to select a MAC codeword $t$, which is further corrupted by artificial noise to form an ANA-MAC codeword $\tilde{t}$. When Eve receives $\tilde{t}$, she can calculate $2^n$ posteriori probabilities  $P(k'|\tilde{t}^r), k' \in \mathcal{K}$, or $2^n$  log-likelihood ratios
\begin{eqnarray}
   l_k(k')=\log \frac{ P(k|\tilde{t})}{ P(k'|\tilde{t})} =  \frac{1}{\sigma_w^2} \sum_{i=1}^{n} \tilde{t}_i [\bar{t}_i(k)-\bar{t}_i(k')], k' \in \mathcal{K}
\end{eqnarray}
Clearly, $l_k(k)=0$. For equidistant MACs with (Hamming) distance $d$, we have that
\begin{eqnarray}
   d_H(t(k),t(k')) = d, \forall k' \neq k.
\end{eqnarray}
Therefore, it is straightforward to compute the mean and variance of $l_k(k')$ for $\forall k' \neq k$ as
\begin{eqnarray}
   E\{l_k(k')\} &=& \frac{2}{\sigma_w^2}d=4\gamma_t d, \nonumber \\
   \text{Var}\{l_k(k')\} &=& \frac{4}{\sigma_w^4}d \sigma_w^2 = \frac{4}{\sigma_w^2}d = 8 \gamma_t d .
\end{eqnarray}

In what follows, we denote $l_k(k')$ by $l_{k'}$ for simplicity. The posteriori probabilities can now be written as
\begin{equation}
    P(k'|\tilde{t}) = e^{-l_{k'}} P(k|\tilde{t}),
\end{equation}
or
\begin{equation}
    P(k'|\tilde{t}) = \frac{e^{-l_{k'}}}{1+\sum_{k'\neq k}e^{-l_{k'}}}.
\end{equation}

\begin{eqnarray*}
    H(K|\tilde{T}) &=&  E\{H(K|\tilde{T}=t)\} \nonumber \\
                   &=&  E\left\{-\sum_{k' \in \mathcal{K}} P(k'|\tilde{t}) \log_2 {P(k'|\tilde{t})} \right\} \nonumber \\
                     &=& E\left\{\log_2\left(1 + \sum_{i=1}^{2^n-1} e^{-l_i}\right)\right\}  \nonumber \\
                             &&+ {\ln(2)}^{-1} \cdot E\left\{\frac{\sum_{i=1}^{2^n-1} l_i e^{-l_i}}{1+\sum_{i=1}^{2^n-1} e^{-l_i}}\right\}.
\end{eqnarray*}

Since $2^n$ is practically very large ($2^{128}$ for $n=128$), the sum of $2^n$ identically-distributed random variables converges to the sum of their mean values, namely,
\begin{eqnarray*}
  \sum_{i=1}^{2^n-1} e^{-l_i} \approx \sum_{i=1}^{2^n-1} E\{e^{-l_i}\}=2^n-1, \\
  \sum_{i=1}^{2^n-1} l_i e^{-l_i} \approx \sum_{i=1}^{2^n-1} E\{l_i e^{-l_i}\} = -4\gamma_t d(2^n-1).
\end{eqnarray*}
Hence, one finally have that
\begin{eqnarray*}
    H(K|\tilde{T})   &\approx& n - \ln(2)^{-1} \frac{4\gamma_t d (2^n-1)}{2^n} \nonumber \\
                     &\approx& n - 4 \ln(2)^{-1} R_c \gamma_b d .
\end{eqnarray*}
\end{proof}

As expected, the conditional equivocation increases when the noise variance increases. For ANA-MACs, one has to consider both the successful authentication probability $1-\alpha$ and the false acceptance probability $\beta$, which, however, is closely related to the noise variance. Therefore, it is of importance to balance these requirements.

\subsection{General Case}
From coding theory, it is well known that the number of codewords for equidistant codes is very limited, which often results into a very low coding rate.

For a binary code $C$ of length $l$ having $s$ distances,  a general result by Delsarte \cite{Delsarte} implies that
\begin{eqnarray}
   |C| \le \sum_{i=0}^s {l \choose i}.
\end{eqnarray}

It should be pointed out that the derivations of (\ref{eq:beta-ed}) and (\ref{eq:happrox}) require the property of distance invariant for the underlying codes, since we cannot assume the use of a particular key between Alice and Bob. Fortunately, Delsarte told us how to decide if a code is distance invariant.

\begin{lem}[Distance Invariant \cite{Delsarte}]
Let $C$ be a code for which the number $s$ of distances is at most equal to the dual distance $d'$. Then $C$ is distance invariant.
\end{lem}

Unfortunately, it still remains a challenge for design of such distant-invariant codes in practice.

For ANA-MACs, the complexity-theoretic MACs are employed, which can be seen as random codes, due to their pseudorandomness property. Empirically, we claim that the complexity-theoretic MACs are distance-invariant thanks to their inherent pseudorandomness, as verified by extensive numerical results shown in Section-VI.

For the set of random codes of rate $R_c$, it is well known that
\begin{equation}
   A_d = {l \choose d}2^{-l(1-R_c)},
\end{equation}
where $A_d$ denotes the mean number of codewords at Hamming distance $d$ from a fixed codeword.

Then, according to (\ref{eq:beta-ed}) and (\ref{eq:beta}), it is straightforward to show that
\begin{eqnarray}
\label{eq:beta}
   \beta =  \sum_{d>0} \frac{A_d}{2^n} Q\left(\sqrt{2\gamma_b n}\left(2\delta_d-(1-\rho)\right)\right),
\end{eqnarray}
while the successful authentication probability (\ref{eq:alpha-ed}) remains unchanged.

\begin{tem}
For an ANA-MAC with the use of $(l, 2^n)$ MACs, the conditional equivocation about the key when the opponent has accessed a single tag can be approximated as
\begin{eqnarray}
\label{eq:eqv-approx}
    H(K|\tilde{T}) \approx  n - 4 \ln(2)^{-1} R_c \gamma_b \cdot \bar{d} .
\end{eqnarray}
where $\bar{d}=\left(2^{-n}\sum_d d A_d\right)$.
\end{tem}
\begin{proof} Let $\mathcal{K}(d)$ denote the set of keys with which the generated tags are at Hamming distance $d$ from the tag with $k$.  Clearly, $\bigcup_{d\ge 0} \mathcal{K}(d) = \mathcal{K}$. Hence,
\begin{eqnarray*}
    H(K|\tilde{T}) &=&  E\{H(K|\tilde{T}=t)\} \nonumber \\
                   &=& E\left\{\log_2\left(1 + \sum_{k'\in \mathcal{K}/k} e^{-l_k'}\right)\right\} \nonumber \\
                    && + \ln(2)^{-1} E\left\{\frac{\sum_{k'\in \mathcal{K}/k} l_k' e^{-l_k'}}{1+\sum_{k'\in \mathcal{K}/k} e^{-l_k'}}\right\},
\end{eqnarray*}
where
\begin{eqnarray*}
     E\left\{1 + \sum_{k'\in \mathcal{K}/k} e^{-l_k'} \right\} &=& 1 + \sum_{d\ge 1} E\left\{ \sum_{k' \in \mathcal{K}(d)}e^{-l_k'}\right\}  \\
       &\approx& 2^n,
\end{eqnarray*}
and
\begin{eqnarray*}
     E\left\{\sum_{k'\in \mathcal{K}/k} l_k' e^{-l_k'} \right\} &=& \sum_{d\ge 1} E\left\{ \sum_{k' \in \mathcal{K}(d)}l_k' e^{-l_k'}\right\}  \\
       &\approx&  \sum_{d\ge 1} A_d (-4\gamma_t d).
\end{eqnarray*}
\end{proof}

\section{Numerical Results}
We consider ANA-MACs, where the underlying MACs are constructed by the Rijndael block cipher \cite{Daemen:2002:DRA}. Hence, the underlying MACs in ANA-MACs allow the specification of variants with the block length ($l$) and key length ($n$) both ranging from 128 to 256 bits in steps of 32 bits.

\subsection{Empirical Distance Distribution of Complexity-Theoretic MACs}

To make sense a channel coding formulation for the  Rijndael-cipher based MACs, we use $n=128$ and $l=256$. Hence, the coding rate is $R_c=1/2$.

Given a $s\in \mathcal{S}$ and further fix a $k \in \mathcal{K}$, it is straightforward to generate authentication tags with $\forall k'\in \mathcal{K}/k$, and the Hamming distance between $\hbar(s,k')$ and $\hbar(s,k)$ can be numerically computed.

\begin{figure}[htb] 
   \centering
   \includegraphics[width=0.5\textwidth]{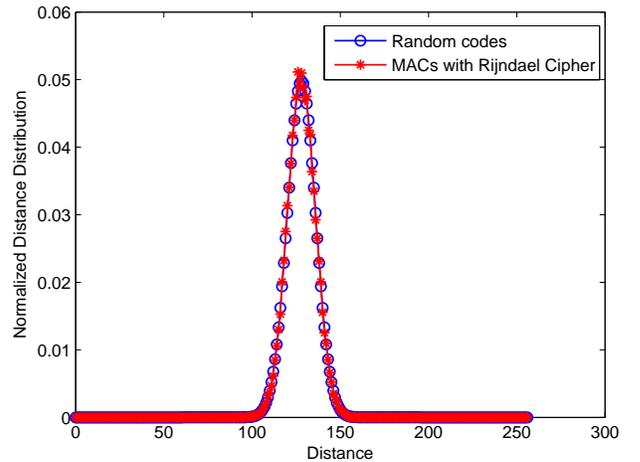} 
   \caption{Distance distribution for random codes and MACs with Rijndael cipher.}
   \label{fig:dis}
\end{figure}

This empirical distance distribution is shown in Fig. \ref{fig:dis}, which coincides well with the random codes of the same coding rate. Extensive numerical results show that this empirical distance distribution keeps unchanged for the use of $\forall s\in \mathcal{S}$ and $\forall k\in \mathcal{K}$. Hence, the distant-invariant property has been empirically confirmed, thanks to the psuedorandomness of the complexity-theoretic MACs.

\subsection{Fundamental Limits on the Key Recovery Attacks}
To attack ANA-MACs, an opponent tries to do her/his best to decode the key.

A fundamental limit on the opponent's capability on guessing the key is the conditional equivocation, $H(K|\tilde{T})$, which can be estimated by (\ref{eq:eqv-approx}).  With a random-code-like distance distribution, it is immediately to see that $\bar{d}=n$. Hence,
\begin{equation*}
 H(K|\tilde{T})= n (1- 4 (\ln 2)^{-1} R_c \gamma_b).
\end{equation*}
Numerically, we, however, found that it is often looser than the lower bound of (\ref{eq:lowb}). This is because that the law of large number holds only approximately when random variables being summed are dependent.

Fig. \ref{fig:eqv-128} shows the lower bound on the normalized conditional equivocation, as determined by (\ref{eq:lowb}). At $E_b/N_0=-3$ dB,  $H(K | \tilde{T}) > 53$.  Hence, the successful probability for an opponent with an unlimited power of computation to guess the key is about $2^{-53}$.
\begin{figure}[htb] 
   \centering
   \includegraphics[width=0.5\textwidth]{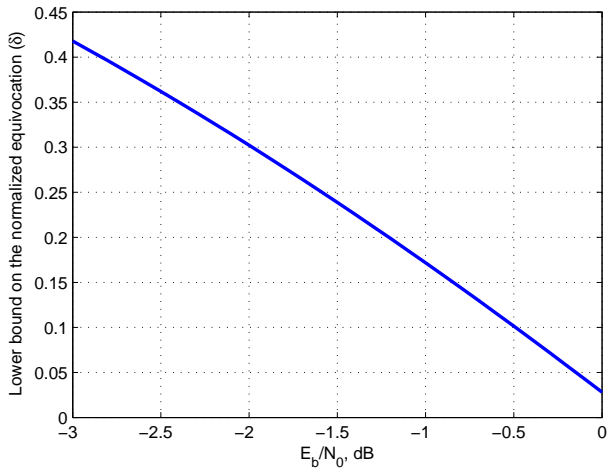} 
   \caption{Lower bound on the normalized equivocation.}
   \label{fig:eqv-128}
\end{figure}

\begin{figure}[htb] 
   \centering
   \includegraphics[width=0.5\textwidth]{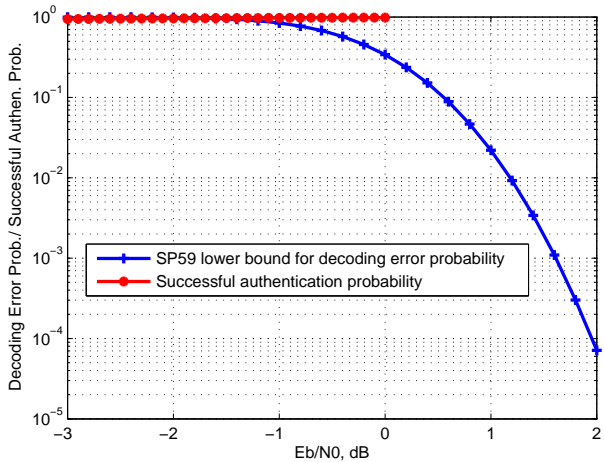} 
   \caption{The SP59 low bound on the decoding error probability and successful authentication probability.}
   \label{fig:spb}
\end{figure}

If the opponent choose to decode the key based on her/his observation of a single authentication tag, we can employ the SP59 lower bound for estimating her/his possibility to successfully decode the key.  Fig. \ref{fig:spb} shows the SP59 bound on the decoding error probability and successful authentication probability for different $E_b/N_0$'s.  As the opponent cannot do better than an ML decoder, the SP59 bound provides an over-estimate of its capability on guessing the key. As shown, the opponent becomes hopeless in guessing the key whenever $E_b/N_0$ is below about -1 dB, where the decoding error probability is around 1, while almost perfect successful authentication probability can still be achieved in this low SNR regime.

\subsection{Completeness Error vs. False Acceptance Probability}
The completeness error $\alpha$ is defined as the complement of the successful authentication probability, which is closely connected to the normalized threshold value $\rho$. By the theory of hypothesis testing, the completeness error and the false acceptance probability $\beta$ is fundamentally balanced with (\ref{eq:alphaBeta}).

\begin{figure}[htb] 
   \centering
   \includegraphics[width=0.5\textwidth]{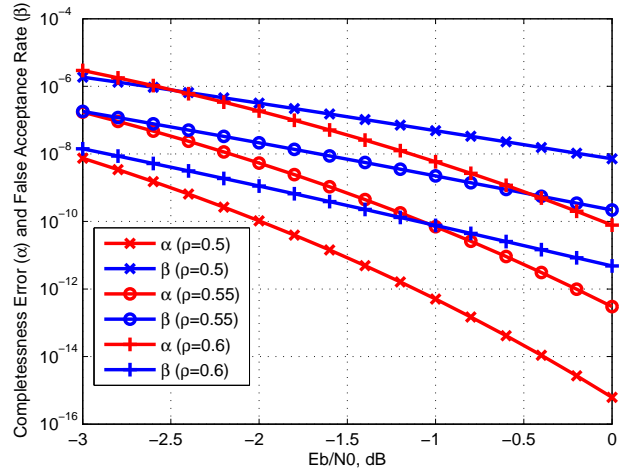} 
   \caption{Completeness error and false acceptance probability versus $E_b/N_0$ for different thresholds $\rho$.}
   \label{fig:alphaBeta}
\end{figure}
To  see the fine tradeoff between $\alpha$ and $\beta$, we plot them in Fig. \ref{fig:alphaBeta} for different $\rho$'s. As the conditional equivocation about the key increases when the SNR decreases, the variance of the artificial noise is essentially determined by the system requirement on the completeness error and false acceptant rate.

\subsection{The Effect of Quantization}
To facilitate packet transmission, quantization should be introduced for ANA-MACs. In most cases, 8-bit quantization is often enough for ANA-MACs and no obvious difference can be observed in simulations for both the successful authentication probability and false acceptance probability with or without quantization. For the conditional equivocation about the key, the introduction of quantization can in general increase it due to the data processing inequality and the opponent becomes more difficult for implementing any key-recovery attack.

\section{Conclusion}
We propose a channel coding approach for the key recovery problem encountered in the spoofing attacks of MACs. With this new approach, the computational security for MACs can be viewed as the requirement of exponential complexity for all possible decoders to succeed.

A new cryptographic primitive, namely, ANA-MACs, is proposed by employing the artificial noise to corrupt the complexity-theoretic MACs. This idea is shown to has some degree of information-theoretic security. The proposed ANA-MACs are similar to the recently-proposed physical layer authentication schemes, as both are interfered with noise. However, the proposed ANA-MACs come with the artificially-introduced noise, the amount of which can be well controlled to meet various performance metrics. This, however, is not the case for physical layer authentication schemes, where the noise is introduced by the channel.

With the introduction of quantization, the proposed ANA-MACs can be encapsulated in packets and transmitted above the physical layer just like that of the traditional MACs, which contrasts sharply with the existing physical layer authentication schemes. We hope that this research can bridge two closely-related but almost independently developed primitives, namely, information-theoretic authentication codes, and complexity-theoretic MACs.


\end{document}